\documentclass[10pt,sigconf]{acmart}
\usepackage{enumerate}
\usepackage{amsmath}
\usepackage{amsthm}
\usepackage{mathtools}
\usepackage{booktabs}
\usepackage{graphicx}
\usepackage{subcaption}
\usepackage{multirow}
\usepackage{adjustbox}
\usepackage{url}
\usepackage{color}
\usepackage{xspace}
\usepackage{xifthen}
\usepackage[noend]{algpseudocode}
\usepackage[linesnumbered,ruled,vlined]{algorithm2e}

\AtBeginDocument{%
  \providecommand\BibTeX{{%
    \normalfont B\kern-0.5em{\scshape i\kern-0.25em b}\kern-0.8em\TeX}}}

\setcopyright{acmcopyright}
\copyrightyear{2021}
\acmYear{2021}
\acmDOI{10.1145/1122445.1122456}

\acmConference[PODS '21]{Proceedings of the 40th ACM SIGMODSIGACT-SIGAI 
Symposium on Principles of Database Systems (PODS’21)}{}{}
\acmBooktitle{Proceedings of the 40th ACM SIGMODSIGACT-SIGAI 
Symposium on Principles of Database Systems (PODS’21)}
\acmPrice{15.00}
\acmISBN{978-1-4503-XXXX-X/18/06}

\newcommand{\blit}[1]{\color{black}#1\color{black}\xspace}
\newcommand{\code}[1]{\mathcal{#1}}
\newcommand{\E}{\mathbb{E}}
\newcommand{\eat}[1]{}
\newcommand{\eps}{\varepsilon}
\newcommand{\INDEX}{\textsf{Index}\xspace}
\newcommand{\mat}[1]{\mathbf{#1}}
\newcommand{\poly}[1]{\operatorname{poly}(#1)}
\newcommand{\prob}{\mathbb{P}}
\newcommand{\QstarFull}[2]{\textsf{star}^{#2}(#1)}
\newcommand{\Qstar}[1]{\textsf{star}(#1)}
\newcommand{\R}{\mathbb{R}}
\newcommand{\ra}[1]{\renewcommand{\arraystretch}{#1}} 
\newcommand{\redit}[1]{#1}
\newcommand{\supp}[1]{\operatorname{supp}(#1)}
\newcommand{\uniform}[1]{\texttt{uSample}{$(#1)$}}
\newcommand{\zerovec}{\mathbf{0}}
\newcommand{\powerset}[1]{\mathcal{P}\left( #1 \right)}
\newcommand{\net}[1][]{%
\ifthenelse{\isempty{#1}}{$\alpha$-net}{$ #1 $-net}%
}
\newcommand{\neighbour}[1][]{%
\ifthenelse{\isempty{#1}}{$\alpha$-neighbour}{$ #1 $-neighbour}%
}
\newtheorem{remark}{Remark} 

\begin{document}

\title[Bounds on Projected Frequency Estimation]{Subspace exploration: Bounds on Projected Frequency Estimation}

\author{Graham Cormode}
\affiliation{%
  \institution{University of Warwick}
}

\author{Charlie Dickens}
\affiliation{%
  \institution{University of Warwick}
}

\author{David P. Woodruff}
\affiliation{%
  \institution{Carnegie Mellon University}
}

\renewcommand{\shortauthors}{Cormode, Dickens, Woodruff}

\begin{abstract}
  Given an $n \times d$ dimensional dataset $A$, a projection query
  specifies a subset $C \subseteq [d]$ of columns which yields a new
  $n \times |C|$ array.
  We study the space complexity of computing data analysis
  functions over such subspaces, including heavy hitters and norms, when the
  subspaces are revealed only after observing the data.
  We show that this important class of problems is typically hard:
  for many problems, we show $2^{\Omega(d)}$ lower bounds.
  However, we present upper bounds
  which demonstrate space dependency better than $2^d$.
  That is, for $c,c' \in (0,1)$ and a parameter $N=2^d$ an $N^c$-approximation can be obtained
  in space $\min(N^{c'},n)$, showing that it is possible to improve on the na\"{i}ve approach
  of keeping information for all $2^d$ subsets of $d$ columns. 
  Our results are based on careful constructions of instances using coding
  theory and novel combinatorial reductions
  that exhibit such space-approximation tradeoffs.
\end{abstract}




\maketitle

\section{Introduction}
  In many data analysis scenarios, datasets of interest
  are of moderate to high dimension, but many of these dimensions are spurious
  or irrelevant.
  Thus, we are interested in subspaces, corresponding to the data
  projected on a particular subset of dimensions.
  Within each subspace, we are concerned with computing
  statistics, such as norms, measures of variation, or finding common
  patterns.
  Such calculations are the basis of subsequent analysis, such as
  regression and clustering.
  In this paper, we introduce and formalize novel problems related to
  functions of the frequency in such projected subspaces.
  Already, special cases such as subspace projected distinct elements
  have begun to generate interest, e.g.,
  in Vu's work \cite{HoaVu2018}, and as an open problem in sublinear
  algorithms \cite{sublinear94}.

  In more detail, we consider the original data to be represented by a
  (usually binary) array with $n$ rows of $d$ dimensions.
  A subspace is defined by a set $C \subseteq [d]$ of columns, which
  defines a new array with $n$ rows and $|C|$ dimensions.
  Our goal is to understand the complexity of answering
  queries, such as which rows occur most frequently in the projected
  data, computing frequency moments over the rows, and so on.
  If $C$ is provided prior to seeing the data, then the projection can
  be performed online, and so many of these tasks reduce to
  previously studied questions.
  Hence, we focus on the case when $C$ is decided \emph{after} the data is
  seen.
  In particular, we may wish to try out many different choices of $C$
  to explore the structure of the subspaces of the data.
  Our model is given in detail in Section \ref{sec: preliminaries}.

  For further motivation, we outline some specific areas where such problems arise.

  \begin{itemize}
  \item
    \textbf{Bias and Diversity.}
    A growing concern in data analysis and machine learning is whether outcomes are `fair'
    to different subgroups within the population, or whether they
    reinforce existing disparities.
    A starting point for this is to quantify the level of bias within
    the data when different features are considered.
    That is, we want to know whether certain combinations of attribute
    values are over-represented in the data (heavy hitters), and
    how many different combinations of values are represented in the data
    (captured by measures like $F_0$).
    We would like to be able to answer such queries accurately for many
    different (typically overlapping) subsets of dimensions.
  \item
    \textbf{Privacy and Linkability.}
    When sharing datasets, we seek assurance that they are not
    vulnerable to attacks that exploit structure in the data to re-identify individuals.
    An attempt to quantify this risk is given in recent
    work~\cite{khll}, which asks how many distinct values occur in the data for each
    partial identifier, specified as a subset of dimensions.
    This prior work considered the case where the target dimensions are known in
    advance, but more generally we would like to compute such measures
    for arbitrary subsets, based on frequency moments and sampling
    techniques.
  \item
    \textbf{Clustering and Frequency Analysis.}
    In the area of clustering, the notion of subspaces has been studied
    under a number of interpretations.
    The common theme is that the
    data may look unclustered in the original space due to spurious
    dimensions inflating the distance between points that are otherwise close.
    Many papers addressed this as a search problem:
    to search through exponentially many subspaces to find those in which
    the data is well-clustered.  See the survey by Parsons, Haque and
    Liu~\cite{Parsons:Haque:Liu:04}.
    In our setting, the problem would be to estimate various measures of
    density or clusteredness for a given subspace.
    A related problem is to find subspaces (or ``subcubes'' in database
    terminology) that have high frequency.
    Prior work proceeded under strong statistical
    independence assumptions about the values in different
    dimensions, for example, that the distribution can be modeled
    accurately with a (Na\"ive) Bayesian model~\cite{Kveton:Muthukrishnan:Vu:Xian:18}.
\end{itemize}

\section{Preliminaries and Definitions} \label{sec: preliminaries}
For a positive integer $Q$, let $[Q] = \{0,1,\ldots, Q-1\}$,
and $A \in [Q]^{n \times d}$ be the input data.
The objective is to keep a summary of $A$ which is used to
estimate the solution to a problem $\mathbf{P}$ upon receiving a column subset
query $C \subseteq [d]$.
Problems $\mathbf{P}$ of interest are described in Section \ref{sec: problem-defs}.
Define the restriction of $A$ to the columns indexed by $C$ as $A^C$
whose rows $A^C_i$, $1 \le i \le n$, are vectors over $[Q]^{|C|}$.
We use the Minkowski norm
$\|X\|_p = (\sum_{i,j} |X_{ij}|^p )^{1/p}$ to
denote the entrywise-$\ell_p$ norm
for vectors ($j=1$) and matrices ($j > 1$).


\medskip
\noindent
\textbf{Computational Model.}
\eat{We will access a matrix through its frequency vector representation in the
following way: for the matrix $A$ and a column query $C$, we ask questions of
the underlying frequency vector $f(A,C) \in [Q]^{|C|}$ whose entries are $f_i(A,C)$.
The value $f_i(A,C)$ is the frequency of string $i \in [Q]^{|C|}$ observed on
$A^C$.}
First, the data $A$ is received under the assumption that it is too large to hold
entirely in memory so can be modeled as a stream of data.
Our lower bounds are not strongly dependent on the order in which the
data is presented.
After observing $A$, a {\it column query $C$} is presented.
The frequency vector over $A$ induced by $C$ is $f = f(A,C)$ whose entries 
$f_i(A,C)$ denote the frequency of $Q$-ary word $w_i \in [Q]^{|C|}$.
We study functions of the frequency vector $f = f(A,C)$
after the observation of $A$ and receiving column query $C$.
The task is, during the observation phase, to design a summary of $A$ which
approximates statistics of $A^C$, the restriction of $A$ to its projected subspace $C$.
Approximations of $A^C$ are accessed through the frequency vector $f(A,C)$.
Note that functions (e.g., norms) are taken over $f(A,C)$ as opposed to the
raw vector inputs from the column projection.

\begin{remark}[Indexing $Q$-ary words into $f$]
\label{rem:index-function}
Recall that the frequency vector $f(A,C)$ has length 
$Q^{|C|}$ with each entry $f_i$ counting the occurrences of 
word $w_i \in [Q]^{|C|}$. 
To clearly distinguish between the (scalar) \textit{index} $i$ of $f$
and the input vectors $w_i$ whose frequency is measured by $f_i$ we 
introduce the \textbf{index function} $e(w_i) = i$.
We may think of $e(\cdot)$ as simply the canonical mapping from 
$[Q]^{|C|}$ into $\{0,1,2,\dots, Q^C -1 \}$, 
but other suitable bijections may be used.
\end{remark}

For example, suppose $Q=2$ and $A \in \{0,1\}^{5 \times 3}$ with column
indices $\{1,2,3\}$ given below.
If $C = \{1,2\}$, then using the canonical mapping from $\{0,1\}^{|C|}$ into $\{0,1,2,3\}$ (e.g
$e(00) = 0, e(01) = 1, \dots e(11) = 3$)
we obtain $A^C$ and hence $f(A,C) = (1,1,0,3)$.
\begin{equation*}
  A =
\begin{bmatrix}
1 & 1 & 0 \\
0 & 1 & 0 \\
0 & 0 & 1 \\
1 & 1 & 1 \\
1 & 1 & 0 \\
\end{bmatrix}
\qquad \longrightarrow \qquad
A^C =
\begin{bmatrix}
1 & 1  \\
0 & 1  \\
0 & 0  \\
1 & 1  \\
1 & 1  \\
\end{bmatrix}
\end{equation*}
The vector $f = f(A,C)$ is then the frequency vector over which we seek to
compute statistical queries such as $\|f\|_0$.
In this example, $\|f\|_0 = 3$ (there are three distinct rows in
$A^C$),
while $\|f\|_1 = 5$ is independent of the choice of $C$.

\medskip
\noindent
\subsection{Problem Definitions.} \label{sec: problem-defs}
The problems that we consider
are
column-projected forms of common streaming
problems (\cite{kane2010optimal}, \cite{braverman2017bptree},
\cite{braverman_et_al:LIPIcs:2018:9411}).
Here, we refer to these problems as ``projected frequency estimation
problems'' over the input $A$.
We define 
\begin{align}
  f_i(A,C) &= |\{j : A_j^{C} = w_i,
                              j \in [n] \}| \\
  F_p(A,C) &= \sum_{i \in \{0,1\}^{|C|}} f_i(A,C)^p.
\end{align}

\begin{itemize}
  \item{\textbf{$F_p$ estimation}:
    Given a column query $C$,
    the $F_p$ estimation problem is to approximate the quantity
    $F_p(A,C) = \|f(A,C)\|_p^p$ under some measure of approximation to be
  specified later (e.g., up to a constant factor).
  Of particular interest to us is
  (projected) $F_0(A,C)$ estimation, which counts the number of distinct
    row patterns in $A^C$. }

  \item{\textbf{$\ell_p$-heavy hitters}:
  The query is specified by
    a column query $C \subseteq [d]$, a choice of metric/norm $\ell_p, p>0$
  and accuracy parameter $\phi \in (0,1)$.
%
  The task is then to identify all patterns $w_i$ observed on $A^C$ for which
  $f_i(A,C) \ge \phi \| f(A,C) \|_p$.
  }
  Such values $w_i$ (or equivalently $i$) are called {\it $\phi$-$\ell_p$-heavy hitters},
  or simply {\it $\ell_p$-heavy hitters} when $\phi$ is fixed.
  We will consider a multiplicative approximation based on a parameter
  $c > 1$, where we require
  that all
  $\phi$-$\ell_p$ heavy hitters are reported, and no items
  with weight less than $(\phi/c) \cdot \|f(A,C)\|_p$ are included. 

\item{\textbf{$\ell_p$-frequency estimation}:
    A related problem is to allow the frequency $f_i(A,C)$ to be estimated
accurately, with error as a fraction of $F_p(A,C)^{1/p}=\|f(A,C)\|_p$, which we
refer to as {\it $\ell_p$ frequency estimation}.
Specifically, for a given $w_i$, return an estimate $\hat{f}_i$ which
satisfies $|\hat{f}_i(A,C) - f_i(A,C)| \le \phi \|f(A,C)\|_p$.}
  
  \item{\textbf{$\ell_p$ sampling}:
  \eat{sample rows of $A_i^C$ according to the distribution}
  The goal of this sampling problem is to sample patterns $w_i$
    according to the distribution
  \eat{$p_i \in (1 \pm \eps) \frac{\|A_i^C\|_p^p}{\|A^C\|_p^p} + \Delta$,
  \[ p_i \in (1 \pm \eps) \frac{\|f_i(A,C)\|_p^p}{\|f(A,C)\|_p^p} + \Delta \]}
  $ p_i \in (1 \pm \eps) \frac{f^p_i(A,C)}{\|f(A,C)\|_p^p} + \Delta $
  where $\Delta = 1/\poly{nd}$, and return a $(1 \pm \eps')$-approximation to
  the probability $p_i$ of the item $w_i$ returned. }

\end{itemize}

When clear, we may drop the dependence upon $C$ in the notation and
write $f_i$ and $F_p$ instead.
We will use $\tilde{O}$ and $\tilde{\Omega}$ notation to supress
factors that are polylogarithmic in the leading term.
For example, lower bounds stated as $\tilde{\Omega}(2^d)$ suppress
terms polynomial in $d$.

\subsection{Related Work}

The model we study is reminscent of, but distinct from, some related
formulations.
In the problem of cascaded aggregates~\cite{Jayram:Woodruff:09}, we imagine the
starting data as a matrix, and apply a first operator (denoted $Q$) on
each row to obtain a vector, on which we apply a second operator $P$.
Our problems can be understood as special cases of cascaded aggregates
where $Q$ is a project-then-concatenate operator, to obtain a vector
whose indices correspond to the concatenation of the projection of a
row. Another example of a cascaded aggregate is a so-called correlated
aggregate \cite{tw12}, but this was only studied in the context of two
dimensions.
To the best of our knowledge, our projection-based
definitions have not been previously studied
under the banner of cascaded aggregates.

Other work includes results on provisioning queries for analytics
\cite{AssadiKLT16}, but the way these statistics are defined is
different from our formulation.
In that setting there are different scenarios (``hypotheticals'') that may or may not
be turned on: this corresponds to ``what-if'' analysis whereby a query is roughly 
``how many items are observed if a given set of columns is present (turned on)?''
The number of distinct elements for the query is the union of the number
of distinct elements across scenarios. In our setting, we concatenate the
distinct items into a row vector and count the number of distinct vectors. Note
that in the hypotheticals setting in the binary case,
each column only has $2$ distinct values, $0$ and $1$,
and thus the union also only has $2$ distinct values. However, we can obtain
up to $2^d$ distinct vectors. Consequently, Assadi et al. are able to achieve
$\poly{d/\eps}$ space for counting distinct elements, whereas we show
a $2^{\Omega(d)}$ lower bound. Moreover, they achieve a $2^{\Omega(d)}$ lower
bound for counting (i.e., $F_1$), whereas we achieve a constant upper bound.
These disparities highlight the differences in our models.

More recently, the notion of ``subset norms'' was introduced by
Braverman, Krauthgamer and Yang~\cite{Braverman:Krauthgamer:Yang:18}.
This problem considers an input that defines a vector $v$, where the
objective is to take a subset $s$ of entries of $v$ and compute the norm.
Results are parameterized by the ``heavy hitter dimension'', which is a measure
of complexity over
the set system from which $s$ can be drawn.
While sharing some properties with our scenario, the results for this
model are quite different.
In particular, in~\cite{Braverman:Krauthgamer:Yang:18} a trivial upper bound follows by maintaining the vector
$v$ explicitly, of dimension $n$.
Meanwhile, many of our results show lower bounds that are exponential
in the dimensionality, as $2^{\Omega(d)}$, though we also obtain non-trivial upper bounds. 

\section{Contributions}


The main challenge here is that the column query
$C$ is revealed {\it after} observing the data;
consequently, applying a known algorithm to just the columns $C$ as
the data arrives is not possible.
For example, consider the exemplar problem of counting the number of
distinct rows under the projection $C$, i.e., the projected $F_0$
problem.
Recall that $A^C_{i}$ denotes the $i$-th row of array $A^C$.
Then the task is to count the number of distinct rows observed in $A^C$, i.e.,
\begin{align*}
  F_0(A,C) &= | \{ A_{j}^C : j \in [n] \} |  
           = \| f(A,C) \|_0.
\end{align*}

Observe that $F_0(A,C)$ can vary widely over different choices of
$C$.
For example, even for a binary input $A \in \{0,1\}^{n \times d}$, $F_0(A,C)$
can be as large as $2^d$ when $C$ consists of all columns from a highly
diverse dataset, and as small as $1$
or $2$ when $C$ is a single column or when $C$ selects homogeneous
columns (e.g., the columns in $C$ are all zeros).

\subsection{Summary of Results}
\eat{For the above problems we obtain the following results:
\begin{itemize}
  \item{For $F_p$ estimation we show space  bounds of
  $2^{\Omega(d)}$ for $p \ne 1$, while $O(\log n)$ suffices for
    for $p=1$.
    Similar lower bounds apply for $\ell_p$ sampling of rows. }

  \item{We show that the $\ell_p$-Heavy hitters can be obtained by sampling
  $\poly{d/\eps}$ rows from $A$ for $p\le 1$, but the problem has no polynomial
    upper bound when $p>1$.}
    
\end{itemize}
}

\blit{
Our main focus, in common with prior work on streaming algorithms, is on space complexity.
For the above problems we obtain the following results:
\begin{itemize}
\item{In Section \ref{sec: f0-bounds} we show that projected $F_0$ estimation
  requires $2^{\Omega(d)}$ space for a constant factor approximation,
  demonstrating the essential hardness of these problems. Nevertheless, we obtain a tradeoff
  in terms of upper bounds described below.
}

\item{Section \ref{sec: l_p-problems} presents results for $\ell_p$
frequency estimation, $\ell_p$ heavy hitters, $F_p$ estimation, and $\ell_p$
sampling.
We show a space upper bound of $O(\eps^{-2} \log(1/\delta))$ for
$\ell_p$ frequency estimation when $0 < p < 1$
and complement this result with lower bounds for heavy hitters when $p>1$,
$F_p$ estimation
and $\ell_p$ sampling for all $p \ne 1$, showing that these problems
require $2^{\Omega(d)}$ bits of space.}
\item
  In Section~\ref{sec:upperbounds}
we show upper bounds for $F_0$ and $F_p$ estimation which improve on
the exhaustive approach of keeping summaries of all $2^d$ subsets of columns,
by showing that we can obtain coarse approximate answers with a
smaller subset of materialized answers.  
Specifically, for parameters $N=2^d$ and $\alpha \in (0,1)$ we can obtain an  
$N^{\alpha}$ approximation in $\min \left( N^{H(1/2 - \alpha)}, n\right)$ space.
Since the binary entropy function $H(x) < 1$, this bound is better 
than the trivial $2^d$ bound.
\end{itemize}
}

These bounds show that there is no possibility of ``super efficient''
solutions that use space less than exponential in $d$.
Nevertheless, we demonstrate some solutions whose dependence is still
exponential but weaker than a na\"ive $2^d$. Thinking of $N = 2^d$, the above upper and lower
bounds imply the actual complexity is a nontrivial polynomial function of $N$. 

The bounds also show novel dichotomies that are not present in comparable
problems without projection.
In particular, we show that (projected) $\ell_p$ sampling is difficult for $p \ne 1$ while
(projected) $\ell_p$-heavy hitters has a small space algorithm for $0<p<1$.
This differs from the standard streaming model in which the (classical) $\ell_p$ heavy hitters problem has a small space solutions for
$p\le 2$ without projection~\cite{Larsen:Nelson:Nguyen:Thorup:16}, and
(classical) $\ell_p$ sampling can be performed efficiently for $p\le
2$~\cite{Jayaram:Woodruff:18}. Our lower bounds are built on
amplifying the frequency of target codewords for a carefully chosen
test word. 

Note that there are trivial na\"ive solutions which simply retain the entire input
and so answer the query exactly on the query $C$:
to do so takes $\Theta(nd)$ space, noting that $n$ may be exponential in $d$.
Alternatively, if we know $t = |C|$ then we may enumerate all ${d \choose t}$
subsets of
$[d]$ with size $t$ and maintain (approximate) summaries for each
choice of $C$.
However, this will entail a cost of at least $\Omega(d^{t})$ and as such
does not give a major reduction in cost.

\subsection{Coding Theory Definitions} \label{sec: code-property}

Our lower bounds will typically make use of a \textit{binary} code
$\mathcal{C}$, constituted of a collection of \textit{codewords},
which are vectors (or strings) of fixed length.
We write $\mathcal{B}(l,k)$ to denote all binary strings of length $l$ and
(Hamming) weight $k$.
We first consider the dense, low-distance
family of codes $\mathcal{C} = \mathcal{B}(d,k)$ but will later use 
more sophisticated randomly sampled codes.
%
{When $k < d/2$, we have
${d \choose k} \ge  \left( {d}/{k}\right)^k$ and when
$k = d/2$, we have
${d \choose d/2} \ge {2^d}/{\sqrt{2d}}$.
 A trivial but crucial property of $\mathcal{B}(d,k)$ is that any two codewords from
  this set can have intersecting $1$s in at most $k - 1$ positions.}

    {
We define the {\em support of a string $y$} as $\supp{y} = \{ i : y_i
  \neq 0\}$, the set of locations where $y$ is non-zero.
We define
\textit{child words} to be the set of new codewords obtained from $\mathcal{C}$
  by generating all $Q$-ary words $z$ with $\supp{z} \subseteq \supp{y}$ for
  some $y \in \mathcal{C}$, and construct them with the \textsf{star} operator
  defined next.
  \begin{definition}[$\textsf{star}^{Q}$ operation, \blit{child words}]
  Let $d$ be the length of a binary word, $k$ be a weight parameter, 
  and suppose $y \in \code{B}(d,k)$.
  Let $M = \supp{y}$.
  We define the function $\QstarFull{y}{Q}$ to be the operation which lifts a
  binary word $y$ to a larger alphabet by generating all the words over alphabet
  $[Q]$ on $M$.
  Formally,
  \begin{align*}
    &\QstarFull{y \in \{0,1\}^{d}}{Q} = \{z : z \in [Q]^d, \supp{z} \subseteq \supp{y} \}
  \end{align*}
  Since the alphabet size $Q$ is often fixed when using this operation, when clear
  we will drop the superscript and abuse notation by writing $\Qstar{y}$.
  \blit{Elements of the set $\QstarFull{y}{Q}$ are referred to as
  {\it child words} of $y$.}
  \end{definition}
  For any $y \in \code{B}(d,k)$, there are $Q^{k}$ words generated by
  $\QstarFull{y}{Q}$.
  When $\Qstar{\cdot}$ is applied to all vectors of a set $U$ then we write
  $\Qstar{U} = \cup_{u \in U}{\Qstar{u} }$.
  For example, if $y \in \{0,1\}^d$ and $Q=2$, then $\QstarFull{y}{Q}$ is simply
  all possible binary words of length $d$ whose support is contained in $\supp{y}$.
  }
  \blit{For the projected $F_0$ problem, the code $\mathcal{C} = \mathcal{B}(d,k)$
  is sufficient.
  However, for our subsequent results, we need a randomly
  chosen code whose existence is demonstrated in Lemma \ref{lem: hh-code-existence}.
  The proof
  follows from a Chernoff bound.

  \begin{lemma} \label{lem: hh-code-existence}
  Fix $\epsilon,\gamma \in (0,1)$ and
  let $\code{C} \subseteq \mathcal{B}(d,\epsilon d)$ be such that
  for any two distinct $x,y \in \code{C}$ we have $|x \cap y| \le  (\epsilon^2 + \gamma) d$.
  \blit{
  With probability at least $1 - \exp( -2d\gamma^2)$
  there exists such a code $\code{C}$ with size $2^{O(\gamma^2 d)}$ 
  }
  instantiated by sampling sufficiently many
  words i.i.d. at random from $\mathcal{B}(d,\epsilon d)$.
  \eat{Such a code $\code{C}$ exists and has size $|\code{C}| = 2^{O(\gamma^2 d)}$.}
  \end{lemma}
  }

\begin{proof}
Let $X$ be the random variable for the number of $1$s in common
between $x$ and $y$ sampled
uniformly at random.
Then the expectation of $X$ is
$  \E[X] = \frac{(\epsilon d)^2}{d} = \epsilon^2 d$
and although the coordinates of $x,y$ are not independent, they are negatively
correlated so we may use a Chernoff bound
(see Section $1.10.2$ of \cite{doerr2020probabilistic} for
self-contained details).
Our aim is to show that the number of $1$s in common between $x$ and
$y$ can be at most $\gamma d$ more than its expectation.
Then, via an additive Chernoff-Hoeffding bound:
\[
\prob(X - \E(X) \ge \gamma d)
\le \exp( -2d\gamma^2).
\]

This is the probability that any two codewords $x$ and $y$ are not too similar,
so
by taking a union bound over the $\redit{\Theta}(|\code{C}|^2)$ pairs of codewords,
the size of the code is $|\code{C}| = \exp(d\gamma^2) = 2^{\gamma^2 d/\ln 2}$.
\end{proof}

\subsection{Overview of Lower Bound Constructions}

Our lower bounds rely upon non-standard reductions to the \INDEX problem
using codes $\code{C}$ defined in Section \ref{sec: code-property}.
These reductions are more involved than is typically found as we need to combine
the combinatorial properties of $\code{C}$ along with the $\Qstar{\cdot}$
operation on Alice's input.
In particular, the interplay between $\code{C}$ and $\Qstar{\cdot}$ must be
understood over the column query $S$ given by Bob, which again relies on
properties of $\code{C}$ used to define the input.

Recall that the typical reduction from \INDEX is as follows:
Alice holds a vector $\mathbf{a} \in \{0,1\}^N$, Bob holds an index
$i \in [N]$ and he is tasked with finding $\mathbf{a}_i$ following one-way
communication from Alice.
The randomized communication complexity of \INDEX is $\Omega(N)$
\cite{kremer1999randomized}.
We adapt this setup for our family of problems, following an approach
that has been used to prove many space lower bounds for streaming
algorithms. 

The general construction of our lower bounds is as follows:
first we choose a binary code $\code{C}$ (usually independently at random)
with certain  properties such as a specific weight and a bounded
number of $1$s in common locations with other words in the code.
In the communication setting,
Alice holds a subset $T \subseteq \code{C}$ while Bob holds a codeword
$y \in \code{C}$ and is tasked with determining whether or not $y \in T$.
Bob can also access the index function (Remark \ref{rem:index-function})
$e(y)$ which simply returns the index or location that $y$ is enumerated in 
$\mathcal{C}$.
\eat{Specifically, $e(u_{j'}) = j'$ for every $u_{j'} \in \mathcal{C}$.}
The corresponding bitstring for the \INDEX problem that Alice holds is
\blit{$\mathbf{a} \in \{0,1\}^{|\code{C}|}$} which has $\mathbf{a}_j = 1$ for every element
$w_j \in T$
(under a
suitable enumeration of $\{0,1\}^d$).
We use the $\Qstar{T}$ operator (defined in Section~\ref{sec: code-property})
to map these strings into an input $A$ for each of the problems
(i.e., a collection of rows of datapoints).
Upon defining the instance, we show that Bob can query a proposed
algorithm for the problem and use the output to determine whether or not Alice
holds $y$.
This enables Bob to return $\mathbf{a}_{e(y)}$, 
which is $1$ if Alice holds $y \in T$ and 
$0$ otherwise.
Hence, determining if $y \in T$ or $y \in \code{C}\setminus T$
solves \INDEX and incurs the lower bound
$\Omega(|\code{C}|)$.
Our constructions of $\code{C}$  establish that
$|\code{C}|$ is exponentially large in $d$.

\section{Lower Bounds for $F_0$}  \label{sec: f0-bounds}

In this section, we focus on the $F_0$ (distinct counting) projected
frequency problem.
The main result in this section is a strong lower bound for the
problem, which is exponential in the domain size $d$.

\label{sec: f0-lower-bound}

\eat{We begin by defining some additional terminology and operators that
will help in the construction of the lower bound.}
We use codes $\code{C} = \mathcal{B}(d,k)$ as defined in Section
\ref{sec: code-property}.

\begin{theorem} \label{thm: lower-bound-deterministic-fixed-query-size}
\eat{Let $Q>k$ be the target alphabet size,
and $k$ be a fixed query size.
There exists a choice of input data $A \in [Q]^{n \times d}$
for the projected $F_0$ problem over \blit{$[Q]^k$} so that if $k < d/2$,
any algorithm
achieving an approximation factor of $|Q|/k$ requires space
$\Omega\big((\frac{d}{k})^k \big)$.
\blit{Let $Q \ge 2 $ be an alphabet size and $k < Q$ be a fixed query size.
Any algorithm which achieves an approximation factor of $Q/k$ requires space:
$\Omega\big((\frac{d}{k})^k \big)$}
\eat{if $k < d/2$,
(ii) $\Omega\left(2^d / \sqrt{d} \big)$ if $k = d/2$.}}

Let $Q \ge 2$ be the target alphabet size
and $k < d/2$ be a fixed query size with $Q > k$.
Any algorithm achieving an approximation factor of $|Q|/k$ for the
projected $F_0$ problem requires space $2^{\Omega(d)}$.
\end{theorem}

\begin{proof}
Fix the code $\mathcal{C} = \mathcal{B}(d, k)$, recalling that any
$x \in \mathcal{C}$ has Hamming weight $k$, and for
distinct $x,y \in \mathcal{C}$ at most $k - 1$ bits are shared in common.
We will use these facts to obtain the approximation factor.

\medskip
\noindent
\eat{\textit{Separation.}}
Obtain the collection of all child words $\mathcal{C}_Q$ from $\mathcal{C}$ by
using $\QstarFull{\cdot}{Q}$ as defined in Section \ref{sec: code-property}.
We will reduce from the \INDEX problem in communication complexity as follows.
Alice has a set of (binary) codewords
$T \subseteq \mathcal{C}$ and
\blit{initializes the input array $A$ for the algorithm with all strings
from the set $\Qstar{T}$.}
Bob has a vector $y \in \mathcal{C}$ and wants to know if $y \in T$ or not.
Let $S = \supp{y}$ so that $|S| = k$ and Bob queries the $F_0$ algorithm on columns
of $A$ restricted to $S$.
First suppose that $y \in T$.
Then Alice holds $y$ so $\Qstar{y}$ is included in \blit{$A$}
\redit{and there must be at}
least $Q^k$ patterns observed.
Conversely, if $y \notin T$, then Alice does not include $y$ in \blit{$A$}.
However, by the construction of $\mathcal{C}$, $y$ shares at most $(k-1)$
1s with any distinct $y' \in \mathcal{C}$.
Thus, the number of patterns observed on the columns
corresponding to $S$ is at most
  ${k \choose k-1} Q^{k-1} = k Q^{k-1}$.

We observe that if we can distinguish the case of $k Q^{k-1}$ from
$Q^k$, then we could correctly answer the \INDEX instance, i.e.,
if we can achieve an approximation factor of $\Delta$ such that:
\begin{equation} \label{eq: approximation-factor}
\Delta = \frac{Q^k}{k Q^{k-1}} = \frac{Q}{k}.
\end{equation}

Any protocol for \INDEX requires communication proportional to the
length of Alice's input vector $\mathbf{a}$, which translates into a space lower bound
for our problem.
\redit{Alice's set $T \subset \code{C}$ defines an input vector
for the \INDEX problem built
using a characteristic vector over all words in $\code{C}$,
denoted by
$\mathbf{a} \in \{0,1\}^{|\code{C}|}$, as follows.
Under a suitable enumeration of
$\code{C} = \{ w_1, w_2, \ldots, w_{|\code{C}|} \}$,
Alice's vector is encoded via $\mathbf{a}_i = 1$ if and only if Alice
holds the binary word $w_i \in T$.
From the separation shown earlier, Bob can determine if Alice
holds a word in $T$, thus solving \INDEX and incurring the
lower bound.}
Hence, space proportional to $|\mathcal{C}| = {d \choose k}$ is necessary.
We use the standard relation 
${d \choose k} \ge \left({d}/{k}\right)^k $ and choose $k = a d/2$ for 
a constant $a \in [0,1)$ from which we obtain 
$|\code{C}| \ge 2^{a d / 2}$
\eat{to obtain the desired bound.}
\redit{ to achieve the stated approximation guarantee.}
\eat{
and any algorithm for the projected $F_0$ algorithm on such inputs requires space
$\Omega \left( {d^{k+1}} / {k^k}  \right)$ to achieve the stated
approximation guarantee.}
\end{proof}

\noindent
Setting $k = ad/2$ allows us to vary the query size and directly understand how 
this affects the size of the code necessary for the lower bound.
For a query of size $k$, the size of the input to the projected $F_0$ problem 
is a $\left(d/k\right)^k  \times d$ array $A$ of total size ${d^{k+1}} / {k^k}$.
Theorem \ref{thm: lower-bound-deterministic-fixed-query-size} is for $k < d/2$.
When $k = d/2$ we can use the tighter bound for the central binomial term
 on the sum of the binomial
coefficients and obtain the following stronger bounds.
The subsequent results use the same encoding as in Theorem
\ref{thm: lower-bound-deterministic-fixed-query-size}. However, at certain points
of the calculations the parameter setttings are slightly altered to obtain
different guarantees.

\begin{corollary} \label{cor: lower-bound-deterministic}
Let $Q \ge d/2$ be an alphabet size and $d/2$ be the query size.
There exists a choice of input data $A \in [Q]^{n \times d}$ such that 
any algorithm achieving approximation factor $2Q/d$ for the
projected $F_0$ problem on the query requires space $2^{\Omega(d)}$.
\end{corollary}

\begin{proof}
Repeat the argument of
Theorem \ref{thm: lower-bound-deterministic-fixed-query-size} with $k=d/2$.
The approximation factor from Equation \eqref{eq: approximation-factor} becomes:
  $\Delta = \frac{2 Q}{d}$.
The code size for \INDEX is 
$|\mathcal{C}| \ge 2^d / \sqrt{2d}$.
Note that $|\code{C}|$ is $2^{\Omega(d)}$ as ${\frac12 \log_2(d)}$ can 
always be bounded above by a linear function of $d$.
\eat{and hence the size of the instance}
\redit{The instance is an array whose rows are the $Q^{d/2}$
child words in $\QstarFull{\code{C}}{Q}$.
Hence, the size of the instance}
to the $F_0$ algorithm is:
$\Theta \left( 2^d Q^{d/2} d^{1/2}  \right)$.
\end{proof}

\noindent
Corollary \ref{cor: constant-factor-approx} follows from Corollary
\ref{cor: lower-bound-deterministic} by setting $Q=d$.

\begin{corollary} \label{cor: constant-factor-approx}
A $2$-factor approximation to the projected $F_0$ problem
\blit{on a query of size $d/2$}
needs space
$2^{\Omega(d)}$ with an instance $A$ whose size is
$\redit{\Theta}(2^d d^{\frac{d+1}{2}})$.
\end{corollary}


Theorem \ref{thm: lower-bound-deterministic-fixed-query-size}
and its corollaries suffice to
obtain space bounds over all choices of $Q$.
However, $Q$ could potentially grow to be very large, which may be unsatisfying.
As a result, we will argue how the error varies for fixed $Q$.
To do so, we map $Q$ down to a smaller alphabet of size $q$ and use this code to
define the communication problem from which the lower bound will follow.
The cost of this is that the instance is a logarithmic factor larger in the
dimensionality.

\begin{corollary} \label{cor: lower-bound-deterministic-small-alphabet}
\blit{Let $q$ be a target alphabet size such that} $2 \le q \le |Q|$.
Let $\alpha = Q \log_q(Q) \ge 1$ and
$d' = d \log_q(Q)$.
There exists a choice of input data $A \in [q]^{n \times d'}$ for
which any algorithm for the projected $F_0$ problem
\blit{over queries of size $d/2$ that} guarantees
error $\tilde{O}(\alpha / d')$ requires space $2^{\Omega(d)}$.
\end{corollary}

\begin{proof}
Fix the binary code $\mathcal{C} = \mathcal{B}(d,d/2)$ and generate all child
words over alphabet $[Q]$ to obtain the approximation factor $\Delta = 2Q/d$ as
in Corollary \ref{cor: constant-factor-approx}.
For every $w \in \mathcal{C}$ there are $Q^{d/2}$ child words so the
child code \redit{$\code{C}_Q}$ now
has size $n=\Theta(2^d Q^{d/2} / \sqrt{d})$ words.
Since $Q$ can be arbitrarily large, we encode it via a mapping to a smaller alphabet \redit{but over a slightly larger
dimension};
specifically, use a function $[Q] \mapsto [q]^{\log_q(Q)}$ which generates
$q$-ary strings for each symbol in $[Q]$.
Hence, all of the stored strings in $\code{C}_Q \subset [Q]^d$ are equivalent to a collection,
$\mathcal{C}_q$ over
$[q]^{d \log_q(Q)}$.
\redit{Although $|\code{C}_Q| = |\code{C}_q|$, words in
$\code{C}_Q$ are length $d$, while the equivalent word
in $\code{C}_q$ has length $d \log_q(Q)$.}
This \redit{collection of words from $C_q$} now defines the instance
$A \in [q]^{n \times d \log_q(Q)}$, each word being a row of $A$.
Taking $\alpha = Q \log_q(Q)$ and $d' = d \log_q(Q)$  results in an
approximation factor of:
\begin{equation}
  \Delta = \frac{2 Q}{d} = \frac{2 \alpha}{d'}.
\end{equation}
Alice's input vector $\mathbf{a}$ is defined by the same code $\code{C}$
and held set $T \subset \code{C}$ as in Theorem \ref{thm: lower-bound-deterministic-fixed-query-size} so we incur the 
same space bound.
Likewise, Bob's test vector $y$ and column query $S$ also remain 
the same as in that theorem.

\eat{
Thus we incur the same $\Omega(2^{d}/\sqrt{d})$ space bound.
Since we have used the same code as in Theorem 
\ref{thm: lower-bound-deterministic-fixed-query-size} $\code{C}$ for the
\INDEX{} problem we obtain the same space bound.

The space lower bound again follows from the communication problem on the 
code $\mathcal{C}$ which encodes Alice's input vector $\mathbf{a}$.
Hence, we incur we incur the $\Omega(2^{d}/\sqrt{d})$ space bound associated
to the \INDEX communication problem from Theorem 
\ref{thm: lower-bound-deterministic-fixed-query-size}
The instance $A$ is taken to be the array whose rows are
words from $\mathcal{C}_q$.
\eat{An algorithm for projected $F_0$ would entail a communication protocol
for \INDEX
To obtain this we require at least the space necessary to generate the
code $\mathcal{C}$ for the \INDEX communication problem, thus incurring the
$\Omega(2^{d}/\sqrt{d})$ space bound.}
}
\end{proof}

\noindent Corollary \ref{cor: lower-bound-deterministic-small-alphabet} says that the
same accuracy guarantee as Corollary \ref{cor: lower-bound-deterministic} can be
given by reducing the arbitrarily large alphabet $[Q]$ to a smaller one over
$[q]$.
However, the price to pay for this is that the size of the instance $A$
increases by a factor of $\log_q(Q)$ in the dimensionality.
These various results are summarized in Table~\ref{tab:comparison}.

\begin{table}[t] \centering
\ra{1.0}
\caption{Comparison of the lower bounds for $F_0$.
Theorem \ref{thm: lower-bound-deterministic-fixed-query-size} uses 
$\code{C} = \mathcal{B}(d,k)$, corollaries use 
$\code{C} = \mathcal{B}(d,d/2)$.}
\label{tab:comparison}
\begin{tabular}[t]{lll}
\toprule
&  Instance $A$ for $F_0$ & Approx. Factor  \\
\midrule
Theorem \ref{thm: lower-bound-deterministic-fixed-query-size}
& $\left(\frac{d}{k}\right)^k  \times d$ over $[Q]$
& $\frac{Q}{k}$ \\

Corollary \ref{cor: lower-bound-deterministic}
& $2^d Q^{d/2} \times d$ over $[Q]$
& $\frac{2Q}{d}$ \\

Corollary \ref{cor: constant-factor-approx}
& $2^d d^{d/2} \times d$ over $[d]$
& $2$ \\

Corollary \ref{cor: lower-bound-deterministic-small-alphabet}
& $2^d Q^{d/2} \times d \log_q Q$ over $[q]$
& $\frac{2Q}{d}$ \\
\bottomrule
\end{tabular}
\end{table}%





\section{$\ell_p$-Frequency Based Problems} \label{sec: l_p-problems}

In this section, we extend the techniques from the previous section to
understand the complexity of projected frequency estimation problems
related to the $\ell_p$ norms and $F_p$ frequency moments
(defined in Section~\ref{sec: problem-defs}).
A number of our results are lower bounds, but we begin with a
simple sampling-based upper bound to set the stage.

\subsection{$\ell_p$ Frequency Estimation}
\label{sec:hhupper}

\blit{
We first focus on the projected frequency estimation problem showing that a
simple algorithm keeping a uniform sample of the rows works for $p < 1$.
The algorithm \uniform{A, C, t, b} first builds a uniform sample of $t$ rows
(sampled with replacement at rate $\alpha=t/n$) from $A$
and
evaluates the absolute frequency of string $b$ on the sample after projection
onto $C$.
Let $g$ be the absolute frequency of $b$ on the subsample.
To estimate the true frequency of $b$ on the
entire dataset from the subsample, we return an
appropriately scaled estimator $\hat{f}_{e(b)} = g/\alpha$ which meets the required bounds
given in Theorem \ref{thm: frequency-est-uniform-sample}, 
recalling that $e(b)$ is the index location associated with the string $b$.
The proof follows by a standard Chernoff bound argument and is given in 
Appendix \ref{sec:uniform-sampling-proof}.

\begin{theorem} \label{thm: frequency-est-uniform-sample}
Let $A \in \{0,1\}^{n \times d}$ be the input data and let $C \subseteq [d]$ be a
given column query.
For a given string $b \in \{0,1\}^{C}$, the absolute frequency of $b$, $f_{e(b)}$,
can be estimated up to $\eps \|f\|_1$ additive error using a uniform sample of size
$O(\eps^{-2} \log(1/\delta))$ with probability at least $1-\delta$.
\end{theorem}
\noindent The same algorithm can be used to obtain bounds for all $0 < p < 1$.
By noting that $\|f \|_1 \le \| f \|_p$ for $0 < p < 1$ we can obtain the following
corollary.

\eat{
\begin{proof}
Let $T = \{i \in [n] : A_i^C = b \}$ be the set of indices on which the
projection onto query set $C$ is equal to the given pattern $b$.
Sample $t$ rows of $A$ uniformly with replacement at a rate $q = t/n$.
Let the (multi)-subset of rows obtained be denoted by $B$ and the matrix formed
from the rows of $B$ be $\hat{A}$.
\eat{Since the rows are sampled iid, the expected size of $B$ is $t$.}
For each $i \in B$, define the indicator random variable $X_i$ which is
$1$ if and only if the randomly sampled index $i$ \blit{satisfies
$A_i^C = b$, which occurs with probability $|T|/n$.}
\eat{is included in the random sample, and is in the
set $T$.
That is, $t_i = 1$ if and only if $i \in \hat{T} := T \cap B$.}
Next, we define $\hat{T} = T \cap B$ so that
$|\hat{T}| = \sum_{i = 1}^t X_i$ and the estimator
$Z = \frac{n}{t}|\hat{T}|$ has $\E(Z) = |T|$.
Finally, apply an additive form of the Chernoff bound:
\begin{align*}
\prob \left( |Z - \E(Z)| \ge \eps n \right) &=
\prob \left( \left| \frac{n}{t}|\hat{T}| - |T| \right| \ge \eps n  \right) \\
      &= \prob \left( \left| |\hat{T}| - \frac{t}{n}|T| \right| \ge \eps t \right) \\
      &\le 2 \exp \left(-\eps^2 t \right).
\end{align*}
Setting $\delta = 2 \exp \left(-\eps^2 t \right)$ allows us to choose
$t = O(\eps^{-2} \log(1/\delta))$, which is independent of $n$ and $d$.
\eat{Treating $\eps$ and $\delta$ as constants, this space bound is also constant.}
The final bound comes from observing that $\|f\|_1 = n, f_b = |T| $ and
$\hat{f}_b = Z$.
\end{proof}
}

\begin{corollary} \label{cor: lp-frequency-est-uniform-sample}
Let $A,b, C$ be as in Theorem \ref{thm: frequency-est-uniform-sample}.
Let $0 < p < 1$.
Then uniformly sampling $O(\eps^{-2} \log(1/\delta))$ rows achieves 
$\left| \hat{f}_{e(b)} - f_{e(b)} \right| \le \eps \|f\|_p$ with probability at least
$1-\delta$.
\end{corollary}


Both Theorem \ref{thm: frequency-est-uniform-sample} and Corollary
\ref{cor: lp-frequency-est-uniform-sample} are stated as if $C$ is given.
However, since the sampling did not rely on $C$ in any way, we can
sample complete rows of the input uniformly prior to receiving the
query $C$, which is revealed after observing the data.
The uniform sampling approach also allows us to identify the
$\ell_p$ heavy hitters in \blit{small space:}
for each item included in the sample (when projected onto column set $C$),
we use the sample to estimate its frequency, and declare those with
high enough estimated frequency to be the heavy hitters.
By contrast, for $p > 1$ we are able to obtain a $2^{\Omega(d)}$ space
lower bound, given in the next section.
}

\subsection{$\ell_p$ Heavy Hitters Lower Bound}

Recall that the objective of (projected) $\ell_p$ heavy hitters is to
find all those rows in $A^C$ whose frequency is at least some fraction
of the $\ell_p$ norm of the frequency distribution of this
projection.
\eat{For the lower bound we need a specific type of binary code
We show the existence of such a code $\code{C}$ now.}
\blit{For the lower bound we need a randomly sampled code as defined in Lemma
\ref{lem: hh-code-existence}.}
The lower bound argument follows a similar outline to the
bound for $F_0$, \blit{although now Bob's query is on the complement of the support of his
test vector $y$ (i.e., $S = [d]\setminus \supp{y}$) rather than $\supp{y}$.}
\redit{Akin to Theorem~\ref{thm: lower-bound-deterministic-fixed-query-size}, we will create a reduction
from the \INDEX problem in communication complexity, and use its
communication lower bound to argue a space lower bound for projected
$\ell_p$ heavy hitters.
The proof will generate an instance of $\ell_p$ heavy hitters based on
encoding a collection of codewords, and
consider in particular the status of the string corresponding to all
zeroes.
We will consider two cases:
when Bob's query string is represented in Alice's set of codewords, then
the all zeros string will be a heavy hitter (for a subset of columns
determined by the query); and when Bob's string is
not in the set, then the all zeros string will not be a heavy hitter.
We begin by setting up the encoding of the input to the \INDEX
instance. }

\begin{theorem} \label{thm: heavy-hitters-lower-bound}
Let $\phi \in (0,1)$ be a parameter and fix $p > 1$.
Any algorithm which can obtain a constant factor approximation 
to the projected $\ell_p$-heavy hitters problem
requires space $2^{\Omega(d)}$.
\end{theorem}

\begin{proof}
Fix $\epsilon > 0$.
Let $\code{C} \subset \mathcal{B}(d,\epsilon d)$ be a code whose words have weight
$\epsilon d$ and any two distinct words $x,y$
have at most $(\epsilon^2 + \gamma)d$ ones in common.
By Lemma \ref{lem: hh-code-existence} such a $\code{C}$ exists and
$|\code{C}| = 2^{\Omega_{\gamma}(d)}$.

Suppose Alice holds a \eat{random} subset $T \subset \code{C}$.
Let $\mathbf{a} \in \{0,1\}^{|\code{C}|}$ be the characteristic vector over all length-$d$
binary strings for which $\mathbf{a}_{e(u)} = 1$ if and only if Alice holds $u \in T$.
Bob holds $y \in \code{C}$ and wants to determine if Alice holds $y \in T$.
Ascertaining whether or not Alice holds $y$ would be sufficient for Bob to
solve \INDEX  and incur the $\Omega(|\code{C}|)$ lower bound.

The input array, $A$, for the $\ell_p$-heavy hitters problem is
constructed as follows.
\begin{enumerate}
  \item{Alice populates $A$ with $2^{\epsilon d}$ copies of
  the length-$d$ all ones vector, $\mat{1}_d$}
  \item{Next, Alice takes $Q=2$ and inserts into $A$ the
  collection $\QstarFull{T}{Q}$, which is the expansion
  of her input strings to all child-words in binary.
  That is, for every $s \in T$, Alice computes all binary strings $x$ of
  length $d$ with $\supp{x} \subseteq \supp{s}$
  and includes these in $A$.
  }
\end{enumerate}
\eat{
 $A$ which is
initialized with $2^{\epsilon d}$ copies of the length-$d$ all ones vector, $\mat{1}_d$.
Then Alice takes $Q=2$ and computes $\QstarFull{T}{Q}$ which is the expansion
of her input strings to all sub-words in binary.
That is, for every $s \in T$, Alice computes all binary strings $x$ of
length $d$ with $\supp{x} \subseteq \supp{s}$
and includes these in $A$.
}
Let $S = [d] \setminus \supp{y}$, so that $|S| = d - \epsilon d = (1-\epsilon)d$.
Without loss of generality we may assume $S = \{ 1,2,\ldots,(1-\epsilon)d \}$ and we
denote the $(1-\epsilon)d$ length vector which is identically $0$ on $S$ by $\zerovec_S$.
Suppose there is an algorithm $\code{A}$ which approximates the
$\ell_p$-heavy hitters problem on a given column query up to a constant
approximation factor.
Bob queries $\code{A}$ for the heavy hitters in the table $A$ under the column query
given by the set $S$, and then uses this information to answer whether or not $y \in T$.

\medskip
\noindent
\textbf{Case 1: $y \in T$.}
If $y \in T$, then we claim that $\zerovec_S$ is a $\phi$-$\ell_p$ heavy hitter
for some constant $\phi$, i.e., $f_{e(\zerovec_S)} \ge \phi \|f\|_p$.
We will manipulate the equivalent condition $f_{e(\zerovec_S)}^p \ge \phi^p F_p$.
Since $y \in T$, the set $\Qstar{y}$ is included in the table $A$ as
Alice inserted $\Qstar{s}$ for every $s$ that she holds.
Consider any child word of $y$, that is, a $w \in \Qstar{y}$.
Since $y$ is supported only on $[d] \setminus S$ and $\supp{w} \subseteq \supp{y}$, every
$w_i = 0$ for $i \in S$.
So $\zerovec_S$ is observed once for every $w \in \Qstar{y}$ and
there are $|\Qstar{y}| = 2^{\epsilon d}$ such $w$.
Hence, $\zerovec_S$ occurs at least $2^{\epsilon d}$ times.

Now that we have a lower bound on the frequency of $\zerovec_S$, it remains to
upper bound the $F_p$ value when $y \in \redit{T}$ so that we are assured
$\zerovec_S$ will be a heavy hitter in this instance.
The quantity we seek is the $F_p$ value of all vectors in $A^S$,
\redit{written $F_p(A,S)$; 
 which we decompose into the contribution
from $\zerovec_S$ present due to $y$ being in $T$, and two
special cases from the block of $2^{\eps d}$ all-ones rows and
`extra' copies of $\zerovec_S$ which are contributed by vectors $y' \ne y$}.
We claim that this $\redit{F_p(A,S)}$ value is at most
$|\code{C}|^{1+p} 2^{\epsilon d + (\epsilon^2 + \gamma)dp} +
3 \cdot 2^{\epsilon pd}$.

First, let $y' \in \code{C}$ with $y' \ne y$ and consider prefixes $z$
supported on $S$ which can
be generated by possible child words from $\Qstar{y'}$.
Since our code requires that $|y' \cap y| \leq (\epsilon^2 + \gamma)d$,
$y'$ can have at most $(\epsilon^2 + \gamma)d$ $1$s located in $\bar{S} =
[d] \setminus S$, and hence must have at least $(\epsilon - \epsilon^2 - \gamma)d$
$1$s located in $S$.
Since $|\Qstar{y'}| = 2^{\epsilon d}$, the number of copies of $z$ inserted
is at most $2^{\epsilon d - (\epsilon d - \epsilon^2 d - \gamma d)} = 2^{\epsilon^2d + \gamma d}$.
This occurs for every $y' \in \code{C}$ so the total number of occurences of $z$
is at most $|\code{C}|2^{(\epsilon^2 + \gamma)d}$.
The contribution to $F_p$ for this scenario is then
$|\code{C}|^p 2^{(\epsilon^2 + \gamma)dp}$.
Observe that each codeword $y'$ generates at most $2^{\epsilon d}$
vectors under the $\Qstar{y'}$ operator, so we have an upper bound of
$|\code{C}| 2^{\epsilon d}$ such vectors generated, with a total
contribution of $|\code{C}|^{1+p}2^{(\epsilon^2p + \epsilon +
  \gamma p)d}$.

Next, we focus on 
the two special vectors to count which have a high contribution
to the $F_p$ value.
Recall that Alice specifically included $\mat{1}_d$ into $A$
$2^{\epsilon d}$ times so the $p$-th powered frequency is exactly $2^{\epsilon pd}$ for this term.
From the above argument, $\zerovec_S$ also has frequency
$2^{\epsilon d}$ from $\Qstar{y}$.
But $\zerovec_S$ is also created at most $2^{(\epsilon^2 + \gamma)d}$
times from each $y'\neq y$ in $T$, giving an additional count of at
most
$|\code{C}|2^{(\epsilon^2 + \gamma)d}$.
Based on our choice of $\epsilon$ and $\gamma$, we can ensure that
this is asymptotically smaller than $2^{\epsilon d}$, and so the total
contribution from these two special vectors is at most $3 \cdot
2^{\epsilon d}$.
So in
total we achieve that $F_p$ is at most
$|\code{C}|^{1+p} 2^{\epsilon d + (\epsilon^2 + \gamma)dp} +
3 \cdot 2^{\epsilon pd}$, as claimed.

Then $\zerovec_S$ meets the definition to be a $\phi$-$\ell_p$ heavy hitter provided
\[ 2^{\epsilon pd} > \phi^p (|\code{C}|^{1+p}2^{\epsilon d + (\epsilon^2+\gamma) pd} +
3\cdot 2^{\epsilon pd}).\]

Assuming $p>1$, and choosing
$\epsilon$ sufficiently smaller than $(p-1)/p$ and
$\gamma$ sufficiently small, we have that
\[ |\code{C}|^{1+p}2^{\epsilon d + (\epsilon^2 + \gamma) pd}
\le 2^{O(\gamma^2 d(1+p)) + \epsilon d + \epsilon (p-1)d + \gamma pd}
\le 2^{\epsilon p d}. \]

Hence, we require
$2^{\epsilon pd} > \phi^p O(2^{\epsilon pd})$, i.e.,
$2^{\epsilon d} > \phi O(2^{\epsilon d})$,
which is satisfied for a suitably small but constant $\phi$.

\medskip
\noindent
\textbf{Case 2: $y \notin T$.}
On the other hand, suppose that $y \notin T$. Then the claim is that $\zerovec_S$
is not a $\phi$-$\ell_p$-heavy hitter.
Now the vector $\zerovec_S$ does not occur with a high frequency because
$\Qstar{y}$ is not included in $A$.
However, certain child words in $\Qstar{T}$ could also generate $\zerovec_S$ when
projected onto $S$ and this is the contribution we need to upper bound.
Again, any codeword $s \in T$ has at least $(\epsilon - \epsilon^2-
\gamma)d$ $1$s present on $S$.
So for a particular $s \in T$, $\zerovec_S$ can occur
$2^{\epsilon^2 d + \gamma d}$ times.
Taken over all $y' \in \code{C}$ for which Alice includes in $A$, the
frequency of $\zerovec_S$ in this case is at most $|\code{C}|2^{\epsilon^2 d + \gamma d}$.
Taking $\eps < 1/3, \gamma < \eps/3$ and using $|\code{C}| = 2^{\gamma^2 d / \ln2}$ 
(Lemma \ref{lem: hh-code-existence}) we have $f_{e(\zerovec_S)} \le 2^{0.72 \eps d}$.
Meanwhile, there are $2^{\epsilon d}$ copies of the string $\mat{1}_d$ inserted into
$A$ meaning that $F_p(A,S) \ge 2^{\epsilon pd}$ and hence $F_p^{1/p}$ is strictly
greater than $f_{e(\zerovec_S)}$.
Hence, $\zerovec_S$ is \emph{not} a $\phi$-$\ell_p$ heavy hitter provided that 
$f_{e(\zerovec_S)}/F_p^{1/p} = 2^{-0.28 \eps d}$ is strictly less than $\phi = 1/4$, 
this is satisfied for suitable $\eps$ and $d$.

\eat{
which is strictly less than $2^{\epsilon d}$  provided $\epsilon < \frac12$ and
$\gamma < \frac{1}{4(1+O(1))}$.
Meanwhile, there are $2^{\epsilon d}$ copies of the string $\mat{1}_d$ inserted into
$A$ meaning that the 
\redit{$F_p(A,S) \ge 2^{\epsilon pd}$.}
Since $f_{e(\zerovec_S)} < \phi 2^{\epsilon d}$ the string $\zerovec_S$ cannot be a
$\phi$-$\ell_p$ heavy hitter.
}
\medskip
\noindent
\textbf{Concluding the proof.}
Bob can use his test vector $y$ and a query $S$ with a constant
factor approximation algorithm $\code{A}$ for the $\ell_p$-heavy hitters problem and distinguish
between the two cases of Alice holding $y$ or not based on whether
$\zerovec_S$ is reported.
As a result, Bob can determine if $y \in T$ and consequently solve \INDEX, thus incurring the $\Omega(|\code{C}|) = 2^{\Omega(d)}$ lower bound.
\end{proof}
\noindent
The instance $A$ is initialized with $2^{\eps d}$ rows of the vector
$\mat{1}_d$ and the child words $\QstarFull{T}{Q}$.
For any $t \in \QstarFull{T}{Q}, |\QstarFull{t}{Q}| = 2^{\eps d}$ so
the size of the instance $A$ is $(|T| + 1)2^{\eps d} \times d$.

\subsection{$F_p$ Estimation} \label{sec: Fp-estimation}

The space complexity of approximating the frequency moments $F_p$ has been widely
studied since the pioneering work of Alon, Matias and
Szegedy~\cite{Alon:Matias:Szegedy:99}.
Here, we investigate their complexity under projection.
For $p = 1$, the frequency is always  \redit{the number $n$ of rows 
in the original instance} irrespective of the column set $C$,
so only one word of space is required.
We therefore devote attention to $p \ne 1$.

The reduction to \INDEX for Theorem
\ref{thm: F_p-estimation-lower-bound} follows a similar outline as Theorem
\ref{thm: heavy-hitters-lower-bound} for $p > 1$.
For $p < 1$, we encode the problem slightly differently,
\blit{closer to that in
Theorem~\ref{thm: lower-bound-deterministic-fixed-query-size}.}
\eat{noting
that the column query for Bob in Theorem \ref{thm: heavy-hitters-lower-bound}
is $S = [d] \setminus \supp{y}$ whereas the subsequent result
will use $S = \supp{y}$ for Bob's test vector $y$.}
Again, the reduction to \INDEX relies on Bob determining whether or not
Alice holds $y$, which for $F_p$ estimation amounts to Bob
evaluating $F_p(A,S)$ and comparing to a threshold value.

\begin{theorem} \label{thm: F_p-estimation-lower-bound}
Fix a real number $p > 0$  with $p \ne 1$.
A constant factor approximation to the projected $F_p$ estimation problem
requires space $2^{\Omega(d)}$.
\end{theorem}

\begin{proof}
For $p > 1$
we begin by noticing that in the proof for Theorem
\ref{thm: heavy-hitters-lower-bound} one can also monitor the $F_p$ value of
the input to the problem rather than simply checking the heavy hitters.
In particular, depending on whether or not Alice holds Bob's test word, $y$,
the projected $F_p$ changes by more than a constant.
Consequently, we invoke the same proof for $F_p$, $p>1$ and obtain the same
$2^{\Omega(d)}$ lower bound.

On the other hand, suppose that $p < 1$.
We assume a code $\code{C} \subset \code{B}(d, \epsilon d)$ with the property that
any distinct $x,x' \in \code{C}$ have $|x \cap x'| \le cd$ for some small constant
$c > \epsilon^2$ (see Lemma~\ref{lem: hh-code-existence}).
Again, Alice holds a subset $T \subseteq \code{C}$ and
inserts $\Qstar{T}$
into the table for the problem $A$.
Throughout this proof we use a binary alphabet so suppress the $Q$
notation from $\QstarFull{\cdot}{Q}$.
Bob holds a test
vector $y \in \code{C}$ and is tasked with determining whether or not Alice
holds $y \in T$.
We distinguish between the cases when Alice holds $y \in T$ or not as follows.
Bob uses $y$ to determine the query column set $S= \supp{y}$ and
will compare against the returned frequency value from the algorithm.

\medskip
\noindent
\textbf{Case 1: $y \not \in T$.}
Consider some $y' \in \code{C} \setminus \{ y \}$.
Since $y$ and $y'$ are both codewords, they can have a $1$ coincident
in at most $cd$ locations.
So if Alice does not hold $y$ then the codewords we need to consider are
all binary words in the code which
have at most $cd$ $1$s in common with $y$ on $S$.
We denote this collection of words by $M$, i.e., the set of binary strings of
length $d$ that have at most $cd$ locations set to $1$.
There are $r$ such vectors, where
$r$ is defined by:
\[
r \triangleq \sum_{i=0}^{cd} {d \choose i} \le cd \cdot {d \choose cd}
                            = O(d) 2^{\Theta(cd)} .
                            \]
The total count of all strings generated by Alice's encoding is
at most $2^{\epsilon d}|\code{C}|$: each string in $\code{C}$ generates
$2^{\epsilon d}$ subwords from the $\Qstar{\cdot}$ operation.
We now evaluate the $\ell_p$-frequency of elements in the set $M$, denoted
$F_p(M)$.
For $p<1$, the value $F_p(M)$ is maximized when every element of $M$ has the same
number of occurrences, $|\code{C}|2^{\epsilon d}/r$.
As there are at most $r$ members of $M$, we obtain
$F_p(M) \le |\code{C}|^p 2^{\epsilon d p} r^{1-p}$.
Recalling the bounds on $|\code{C}|$ and $r$, this is:
\begin{equation} \label{eq: Fp-value-to-bound}
2^{cdp + \epsilon dp + \Theta( (1-p)cd )}\cdot O(d^{1-p}).
\end{equation}

We can now choose $c$ to be a small enough constant so that
\eqref{eq: Fp-value-to-bound} is at most $2^{(1-\alpha)\epsilon d}$ for a
constant $\alpha > 0$ by Lemma~\ref{lem: c-bound-Fp} in Appendix
\ref{sec: Fp-estimation-bound}.

\medskip
\noindent
\textbf{Case 2: $y \in T$.}
Now consider the scenario when $y \in T$ so that Alice has inserted $\Qstar{y}$
into the table $A$.
Here, we can be sure that each of the $2^{\epsilon d}$ strings in
$\Qstar{y}$ appears at least once over the column set $S$, and so the
$F_p$ value is at least $2^{\epsilon d} 1^p = 2^{\epsilon d}$.

\medskip
We observe that these two cases
obtain the constant factor separation, as required.
Then, Bob can use his test vector $y$ and a query $S$ with a constant
factor approximation algorithm to the projected $F_p$-estimation problem and
distinguish between the two cases of Alice holding $y$ or not.
Thus, Bob can determine if $y \in T$ and consequently solve the \INDEX{}
problem,  incurring the $\Omega(|\code{C}|) = 2^{\Omega_c(d)}$ lower bound
for a $c$ arbitrarily small.
%
%
\end{proof}

\begin{remark} \label{rem: Fp-instance}
For $p > 1$ we adopt the same instance as in Theorem
\ref{thm: heavy-hitters-lower-bound} so the instance is of size
$(|T| + 1)2^{\eps d} \times d$.
On the other hand, for $0 < p < 1$, only the words in $\QstarFull{T}{Q}$
are required so $A$ has size $|T|2^{\eps d} \times d $.
\end{remark}

\subsection{$\ell_p$-Sampling}
In the projected $\ell_p$-sampling problem, the goal is to sample a row
in $A^C$ proportional to the $p$-th power of its number of occurrences.
One approach to the
standard (non-projected) $\ell_p$-sampling problem on a vector $x$ is
to subsample and find the $\ell_p$-heavy hitters \cite{Larsen:Nelson:Nguyen:Thorup:16}. Consequently,
if one can find $\ell_p$-heavy hitters for a certain value of $p$, then
one can perform $\ell_p$-sampling in the same amount of space, up to
polylogarithmic factors.
Interestingly, for projected $\ell_p$-sampling,
this is not the case, and we show for every $p \neq 1$, there is a
$2^{\Omega(d)}$ lower bound.
This is despite the fact that we can estimate
$\ell_p$-frequencies efficiently for $0 < p < 1$, and hence find the
heavy hitters (Section~\ref{sec:hhupper}).

\begin{theorem}
Fix a real number $p > 0$ with $p \neq 1$, and let $\eps \in (0,1/2)$.
Let $S \subseteq [d]$ be a column query and $i$ be a pattern observed on
the projected data $A^S$.
Any algorithm which returns a pattern $i$ 
sampled from a distribution $(p_1, \ldots, p_n)$, where
$p_i \in (1 \pm \eps) \frac{f_{e(i)}^p}{\|f(A,S)\|_p^p} + \Delta$ together with a
$(1 \pm \eps')$-approximation to $p_i$, $\Delta = 1/\poly{nd}$ and
$\eps' > 0$ is a sufficiently small constant, requires $2^{\Omega(d)}$ bits of space.
\end{theorem}

\begin{proof}
\noindent
\textbf{Case 1: $p>1$.}
The proof of Theorem \ref{thm: heavy-hitters-lower-bound}
argues that the vector $\zerovec_S$ is a constant factor
$\ell_p$-heavy hitter for any $p > 1$ if and only if Bob's test vector $y$ is in
Alice's input set $T$, via a reduction from \INDEX.
That is, we argue that there are constants
$C_1 > C_2$ for which if $y \in T$, then $f_{e(\zerovec_S)}^p \geq C_1 F_p$, while
if $y \notin T$, then $f_{e(\zerovec_S)}^p < C_2 F_p$.
Consequently, given an $\ell_p$-sampler with the guarantees as
described in the theorem statement,
then the (empirical) probability of sampling the item
$\zerovec_S$ should allow us to distinguish the two cases.
This holds even tolerating the $(1+\eps')$-approximation in
sampling rate, for a sufficiently small
constant $\eps'$.
In particular,
if $y \in T$, then we will indeed sample $\zerovec_S$ with $\Omega(1)$
probability, which can be amplified by independent repetition;
whereas, if $y\notin T$, we do not expect to sample $\zerovec_S$ more
than a handful of times.
Consequently, for $p>1$, an $\ell_p$-sampler can be used to solve the
$\ell_p$-heavy hitters problem with arbitrarily large constant
probability, and thus requires $2^{\Omega(d)}$ space.

\medskip
\noindent
\textbf{Case 2: $0 < p < 1$.}
We now turn to $0 < p < 1$.
In the proof of Theorem \ref{thm: F_p-estimation-lower-bound}, a
reduction from \INDEX is described where Alice holds the
set $T$ and Bob the string $y$.
Bob can generate the set $\Qstar{y}$ of size $2^{\eps d}$ which is
all possible binary strings supported on the column query $S$.
From this, Bob constructs the set
$M' = \left\{ z \in \Qstar{y} : |\supp{z}| \ge \frac{\eps d}{2} \right\}$.
We observe that if $y \in T$ then at least half of the strings in
$\Qstar{y}$ are supported on at least $\eps d/2$
coordinates which implies $|M'| \ge 2^{\eps d - 1}$.
The total $F_p$ in this case can be bounded by a contribution of
$|M'|1^p + 2^{\eps d}$.
The first term arises from the $|M'|$ strings in $M'$ with a frequency
of $1$, while the second term is shown in Case 1 of Theorem~\ref{thm:
  F_p-estimation-lower-bound}.
Since $|M'| \leq 2^{\eps d}$, we have that $F_p \leq 2^{\eps d
  + 1}$ in this case.
Consequently, the correct probability of $\ell_p$-sampling returning a
string in $M'$ is at least $\frac14$ for the ``ideal'' case of $\eps
=0, \Delta=0$.
Even allowing $\eps < \frac12$ and
$\Delta=1/\operatorname{poly}({nd})$, this probability is at least
$1/10$.

Otherwise, if $y \not \in T$,
we exploit that $y' \ne y$ can coincide in at most $cd =
O(\eps^2 d)$ coordinates
and $|\supp{z}| \ge \eps d/2 > cd$ for any $z \in M'$.
Hence, no $z \in M'$ can occur in $\Qstar{y'}$ for another
$y' \in \code{C} \setminus \{y\}$ on the column projection $S$.
In this case, there should be zero probability of sampling a string in
$M'$ (neglecting the trivial additive probability $\Delta$).

To summarize,
in the case that $y \in T$, by querying the projection $S$
then a constant fraction of the $F_p$-mass is on the set $M'$,
whereas when $y \notin T$, then there is zero $F_p$-mass on the set $M'$.
Since Bob knows $M'$, he can run an $\ell_p$-sampler and check if the output
is in the set $M'$, and succeed with constant probability.
It follows that Bob can solve the \INDEX problem (amplifying success
probability by independent repetitions if needed), and thus again
the space required is $2^{\Omega(d)}$.
\end{proof}

\begin{remark} \label{rem: lp-sampling-instance}
For $p > 1$ we again adopt the same instance as in Theorem
\ref{thm: heavy-hitters-lower-bound} which has size
$(|T| + 1)2^{\eps d} \times d$.
However, for $0 < p < 1$, we require the instance from Theorem
\ref{thm: F_p-estimation-lower-bound}
so $A$ has size $|T|2^{\eps d} \times d $.
\end{remark}

\section{Projected Frequency Estimation via Set Rounding}
\label{sec:upperbounds}
Although our lower bounds rule out the possibility of computing constant factor
approximations to projected frequency problems in sub-exponential space, it is 
still possible to compute non-trivial approximations using exponential space but still
better than nai\"vely enumerating all column subsets of $[d]$.
We design a class of algorithms that proceed by keeping
appropriate sketch data structures for a ``net'' of subsets.
The net has the property that for any query $C \subset [d]$ there is a
$C' \subset [d]$ stored in the net which is not too different from
$C$.
We can then answer the query on $C$ using the summary data structure
computed for columnset $C'$.
To formalize this approach we need some further definitions, the first
of which conceptualizes the notion of a net over subsets.

\begin{definition}[$\alpha$-net of subsets]
    Let $\powerset{[d]}$ denote the power set of $[d]$.
    Fix a parameter $\alpha \in (0,1/2)$.
    An $\alpha$-net of $\powerset{[d]}$ is the set 
    $\code{N} = \{U : |U| \le 2^{d/2 - \alpha d} \text{~or~} |U| \ge 2^{d/2 + \alpha d} \}$
    which contains all subsets $U \in \powerset{[d]}$ whose size is at most $2^{d/2 - \alpha d}$
    or at least $2^{d/2 + \alpha d}$.
\end{definition}
Let $H(x) = - x \log_2(x) - (1-x) \log_2(1-x)$ denote the {\it binary entropy function}.
\begin{lemma}
    \label{lem:net-size}
    Let $\code{N}$ be an $\alpha$-net for $\powerset{[d]}$.
    Then $|\code{N}| \le 2^{H(1/2 - \alpha)d + 1}$.
\end{lemma}

\begin{proof}
    The total number of subsets whose size is at most 
    $2^{d/2 - \alpha d} $ is $\sum_{i\le \alpha d} {d \choose i}$
    and  $\sum_{i\le \alpha d} {d \choose i} \le 2^{H(1/2 - \alpha)d}$
    \cite[Theorem 3.1]{galvin2014three}.
    By symmetry we obtain the same bound for the number of subsets 
    of size at least $2^{d/2 + \alpha d}$,
    yielding the claimed total.
\end{proof}

\eat{
\begin{definition}[$\alpha$-neighbor]
    Let $C \in \powerset{[d]}$ and $\code{N}$ be an $\alpha$-net.
    If $C \notin \code{N}$ then a {\it $\alpha$ neighbour} is any $C' \in \code{N}$
    which satisfies
    $| C \triangle C'| \leq \alpha d$. 
%
\end{definition}
}

\subsection{From \net{}s to Projections}

\begin{algorithm}[t]
    \KwIn{Data $A \in \{0,1\}^{n \times d}$, 
    parameter $\alpha \in (0,1/2)$, 
    frequency estimation problem $P$,
    query $C$ revealed after $A$}
    \SetKwFunction{FMain}{ProjectedFreq}
    \SetKwProg{Fn}{Function}{:}{}
    \Fn{\FMain{$A,\alpha, C$}}{
        Generate an $\alpha$-net $\code{N}$ \\
        For every $U \in \code{N}$ evaluate a $\beta$ approximate sketch to 
        estimate $P(A,U)$ \\
        Given a projection query $C$ after observing $A$: \\
        Obtain $C'$, an \neighbour{} to $C$ in $\code{N}$ \\
        \Return{$P(A,C')$ to $\beta$ relative error}
    }
     \caption{Projected frequency by query rounding}
     \label{alg:projected-freq-rounding}
\end{algorithm}

\begin{figure*}[ht]
    \centering
    \includegraphics[width=\textwidth]{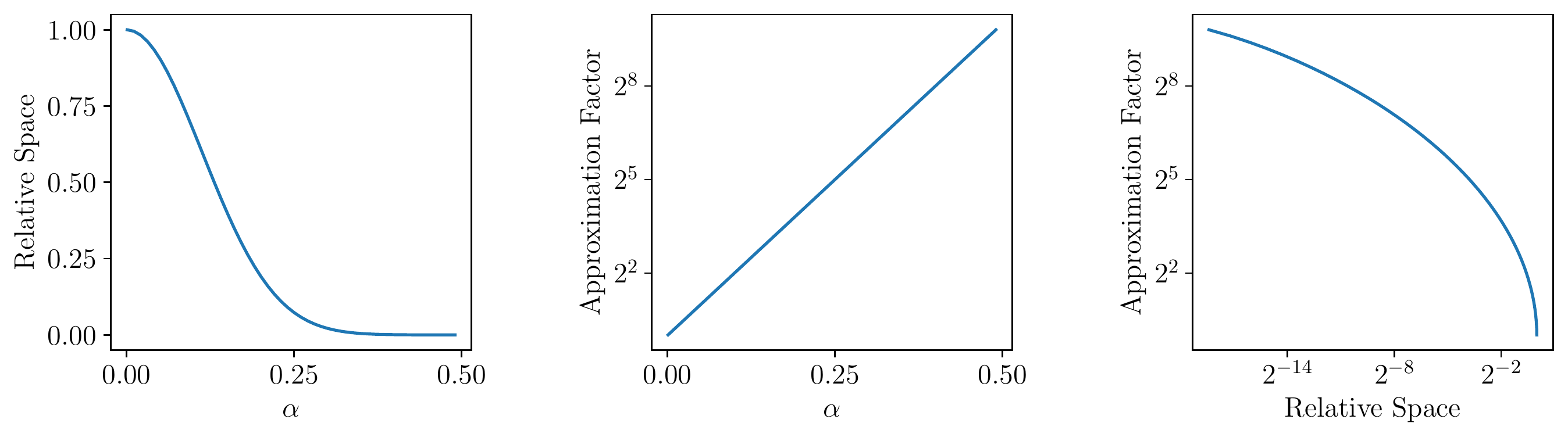}
    \caption{Space-approximation tradeoff for $d=20$
    as $\alpha$ is varied from $0$ to $1/2$.
    Relative space is $2^{H(1/2 - \alpha)d} / 2^d$.}
    \label{fig:space-approx-tradeoff}
\end{figure*}

Suppose that we are tasked with answering problem $P = P(A,C)$ on a projection query $C$.
We know that if $C$ is known ahead of time then we can encode the input data
$A \in [Q]^{n \times d}$ on projection $C$ as a standard stream over the alphabet $[Q]^{|C|}$.
The use of \net{}s allows us sketch some of the input and use this to approximately 
answer a query.
For a standard streaming problem, we will say that an algorithm yields
a \emph{
$\beta$-approximation} to the true solution $z^{*}$ if the returned
estimate $z \in [z^*/\beta, \beta z^*]$.
A sketch obtaining such approximation guarantees will be referred to
as a $\beta$ approximate 
sketch.
We additionally need the following notion of error due to the distortion incurred when answering
queries on elements of the \net{} rather than the given query.

\begin{definition}[Rounding distortion]
    \label{def:rounding-distortion}
    Let $P = P(A,C)$ be a projection query for the problem $P$ on input 
    $A \in [Q]^{n \times d}$ with projection $C$.
    Let $\code{N} \subset \powerset{[d]}$ be an \net{}.
    The \emph{rounding distortion $r(\alpha,P)$} is the worst-case determinstic error 
    incurred by solving $P(A,C')$ rather than $P(A,C)$ for an $\alpha$-neighbour 
    $C' \in \code{N}$ of $C$ so that $P(A,C)/r(\alpha,P) \le P(A,C') \le r(\alpha,P) P(A,C)$.
\end{definition}

Definition \ref{def:rounding-distortion} is easiest to conceptualize for the $F_0$ problem
when $A \in \{0,1\}^{n \times d}$.
Specifically, $P=F_0$ and the task to solve is $P = F_0(A,C)$.
For a given query $C$, with an \neighbour{} $C'$ in the net,
the gap between the number of distinct items 
observed on $C'$
at most doubles for each column in the set difference between $C$ and
$C'$.
Since $C'$ is an \neighbour{}, we have $|C'\operatorname{\Delta} C| \le \alpha d$ so the worst-case 
approximation factor in the number of distinct items observed over
$C'$ rather than $C$ is $2^{\alpha d}$.

More generally, we can categorize the rounding distortion for other typical queries, as 
demonstrated in the following lemma.
Note that if the query is contained in the \net{} $\code{N}$ then we
will retain a sketch for that problem; hence the distortion is only
incurred for queries not contained in the net. 

\begin{lemma}
    \label{lem:rounding-distortion}
    Fix $\alpha \in (0,1/2)$, suppose $A \in \{0,1\}^{n \times d}$ and $\code{N}$ be an \net{}.
    If $C$ is a projection query for the following cases, the rounding distortion can be bounded as:
    \begin{enumerate}   
        \item{$P = F_0(A,C)$ then $r(\alpha, F_0) = 2^{\alpha d}$} \label{item:rounding-f0}
        \item{$P = F_p(A,C), p > 1$ then $r(\alpha, F_p) = 2^{\alpha d (p-1)}$} \label{item:rounding-fp-gt1}
        \item{$P = F_p(A,C), p <  1$ then $r(\alpha, F_p) =  2^{\alpha d (1-p)}$} \label{item:rounding-fp-lt1}
    \end{enumerate}
\end{lemma}

\begin{proof}
    Item \eqref{item:rounding-f0} is an immediate consequence of the
    discussion above following Definition \ref{def:rounding-distortion}
    so we focus on \eqref{item:rounding-fp-gt1} and \eqref{item:rounding-fp-lt1}.
    Suppose $p \ge 1$.
    Let $f_C = f(A,C)$ denote the frequency vector associated to the projection query $C$ over domain $[2^{|C|}]$. 
    First, consider a single index $j \in [2^{|C|}]$ with $(f_C)_j = x$.
    Let $C'$ be an \neighbour{} for $C$ in $\code{N}$, and without
    loss of generality, assume that $|C| < |C'|$.
    The task is to estimate $\|f_C\|_p^p = x^p$
from $\|f_{C'}\|_p^p$, where $f_{C'} = f(A,C')$ is a frequency 
    vector over the domain $[2^{|C'|}]$ which is a $|C'\setminus C|$ factor larger than the domain for $f_C$.
    However, observe that in $f_{C'}$, the value of $x$ is spread
    across the at most $2^{\alpha d}$ entries that agree with $j$
    on columns $C$.
    The contribution to $F_p$ from these entries is at most $x^p$ (if
    the mass of $x$ is mapped to a single entry).
    On the other hand, by Jensen's inequality,
    the contribution is at least $2^{\alpha d}(x/2^{\alpha d})^p =
    x^p/2^{\alpha d(p-1)}$.  
    Hence, considering all entries $j$, we obtain
%
%
%
    $\| f_{C}\|_p^p/2^{\alpha d(p-1)} \le \| f_{C'}\|_p^p \le \| f_{C}\|_p^p$.
    In the case $|C| > |C'|$, essentially the same argument shows that 
$\| f_{C}\|_p^p \le \| f_{C'}\|_p^p \le \| f_{C}\|_p^p 2^{\alpha d(p-1)}$.
    Thus we obtain the rounding distortion of $2^{\alpha d(p-1)}$. 
    For $p < 1$, we proceed as above, except by concavity, the ordering is reversed.
\end{proof}

Observe that the distortion reduces to $1$ (no distortion) as we
approach $p=1$ from either side.
This is intuitive, since the $F_1$ problem is simply to report the
number of rows in the input, regardless of $C$, and so the problem
becomes ``easier'' as we approach $p=1$. 

With these properties in hand, we can give a ``meta-algorithm'' as described in 
Algorithm \ref{alg:projected-freq-rounding}.
In Theorem \ref{thm:projected-freq-net} we can fully characterize the accuracy-space 
tradeoff for Algorithm \ref{alg:projected-freq-rounding} as a function of $\alpha$ and $d$.

\begin{theorem}
    \label{thm:projected-freq-net}
    Let $A \in \{0,1\}^{n \times d}$ be the input data and $C \subseteq [d] $ be a projection query.
    Suppose $P = P(A,C)$ is the projected frequency problem, $\alpha \in (0,1/2)$ and 
    $r(\alpha, d) $ is the rounding distortion.
    With probability at least $1-\delta$
    a $\beta r(\alpha, d)$ approximation can be obtained
    by keeping $\tilde{O}(2^{H(1/2 - \alpha)d})$ $\beta$-approximate
    sketches.
\end{theorem}

\begin{proof}
Let $\code{N}$ be a \net{} for $\powerset{[d]}$ and for every $U \in
\code{N}$ generate a sketch with accuracy parameter $\epsilon$ 
for the problem $P$ on the projection defined by $U \subseteq [d]$.
Either the projection $C \in \code{N}$, in which case we can report a
$\beta$ factor approximation, or 
$C \notin \code{N}$ in which case we take an \neighbour{}, $C' \in \code{N}$ and return the estimate $z$ for $P(A,C')$.
The sketch ensures that the answer to $P(A,C')$ is obtained with
accuracy $\beta$, 
which
by the rounding distortion
is a $\beta r(\alpha, d)$ approximation. 
To obtain this guarantee we
build one 
sketch for every $U \in \code{N}$, for a total of
$O(2^{H(1/2 - \alpha)d})$ sketches (via Lemma \ref{lem:net-size}).
By setting the failure probabilty for each sketch as $\delta=1/2^{\alpha d}$ and then taking
a union bound over the \net{} we achieve probability at least $1-\delta$.
\end{proof}

We remark that similar results are possible for the other functions
considered, $\ell_p$ frequency estimation, $\ell_p$ heavy hitters and $\ell_p$ sampling.
The key insight is that all these functions depend at their heart on
the quantity $f_j/\|f\|_p$, the frequency of the item at location $j$
divided by the $\ell_p$ norm.
If we evaluate this quantity on a superset of columns,
then both the numerator and denominator may shrink or grow, in the
same ways as analyzed in Lemma~\ref{lem:rounding-distortion}, and
hence their ratio is bounded by the same factor, up to a constant.
Hence, we can also obtain (multiplicative) approximation algorithms 
for these problems with similar behavior.

\medskip
\noindent
\textbf{Illustration of Bounds.}
First, observe that, irrespective of the problem $P$, the number of
sketches needed is sublinear in $2^d$.
This is due to the fact that the entropy $H(1/2 - \alpha) < 1$ for
$\alpha>0$, so the size of the net $|\code{N}| < 2^d$.
For $0 \leq p \leq 2$, we have $\beta$-approximate sketches with
$\beta=(1+\epsilon)$ whose size is $\tilde{O}(\eps^{-2})$, which is
constant for constant $\epsilon$. 
For example, we obtain a $2^{\alpha d}$ approximation (ignoring small
constant factors) for $F_0$ in space
$O(2^{H(1/2 - \alpha)}d)$, using for instance the
$(1+\epsilon)$-approximate sketch from 
\cite{kane2010optimal} which requires $O(\eps^{-2} + \log n')$ bits for an input over domain $\{1,\dots,n'\}$.
Since $n' \leq 2^d$, and setting $\epsilon = 1$, we obtain the
approximation in space $O(d 2^{H(1/2 - \alpha)d})$.
This is to be compared to the bounds in Section~\ref{sec: f0-bounds},
where it is shown that (binary) instances of the projected $F_0$
problem require space $2^{\Omega(d)}$. 
These results show that the constant hidden by the $\Omega()$ notation
is less than $1$.

In Figure \ref{fig:space-approx-tradeoff} we illustrate the general
behavior of the bounds for $d=20$.
We plot the \emph{relative space} by $2^{H(1/2 - \alpha)}/2^d$ while varying $\alpha$ over 
$(0,1/2)$ (plotted in the leftmost pane).
This shows the space reduction in using the \net{} approach compared to nai\"vely storing 
all $2^d$ queries.
The central pane shows how the approximation factor $2^{\alpha d}$ (on
a log scale) varies with $\alpha$.
We plot the space-approximation tradeoff in the rightmost pane and the approximation factor is 
again plotted on a $\log_2$-scale.
This plot suggests that if we reduce the space by a factor of $4$ (i.e., permit relative space $2^{-2}$)
then the approximation factor is on the order of 10s.
Meanwhile, if we use relative space $2^{-8}$, then the approximation remains on the order of hundreds:
this is a substantial saving as the number of summaries kept for the approximation is $2^{12} = 4096 \ll 2^{20} \approx 10^6$.

\redit{
\section{Concluding Remarks}

We have introduced the topic of projected frequency estimation, with
the aim of abstracting a range of problems involving computing
functions over projected subspaces of data.
Our main results show that these problems are generally hard, in terms
of the space requirements: in most cases, we require space which is
exponential in the dimensionality $d$ of the input.
However, interestingly, the exact dependence is not as simple as $2^d$:
we show that coarse approximations can be obtained whose cost is
substantially sublinear in $2^d$. Letting $N = 2^d$, our upper and
lower bounds establish that the space complexity for a number of problems
here is polynomial in $N$, though substantially sublinear. 
And, in a few special cases ($\ell_p$ frequency estimation for $p \leq
1$), a sufficiently constant-sized sample suffices 
for accurate approximation of projected frequencies.
It remains an intriguing open question to close the gaps between the 
upper and lower bounds,
and to find the exact form of the polynomial dependence on $N$ for
these problems. 
}

\paragraph{Acknowledgements.}
We thank S. Muthukrishnan and Jacques Dark for helpful discussions
about this problem.
The work of GC and CD was supported by European Research Council grant ERC-2014-CoG 647557.
The work of DW was supported by NSF grant No. CCF-1815840, National Institute of Health grant 5R01HG 10798-2, and a Simons Investigator Award. 

\newpage
\bibliographystyle{ACM-Reference-Format}
\bibliography{sample-base}

\appendix

\section{Omitted Proofs}

\subsection{Omitted Proof for Section~\ref{sec:hhupper}}
\label{sec:uniform-sampling-proof}

\begin{theorem}[Restated Theorem \ref{thm: frequency-est-uniform-sample}] 
Let $A \in \{0,1\}^{n \times d}$ be the input data and let $C \subseteq [d]$ be a
given column query.
For a given string $b \in \{0,1\}^{C}$, the absolute frequency of $b$,
$f_{e(b)}$,
can be estimated up to $\eps \|f\|_1$ additive error using a uniform sample of size
$O(\eps^{-2} \log(1/\delta))$ with probability at least $1-\delta$.
\end{theorem}

\begin{proof}
Let $T = \{i \in [n] : A_i^C = b \}$ be the set of indices on which the
projection onto query set $C$ is equal to the given pattern $b$.
Sample $t$ rows of $A$ uniformly with replacement at a rate $q = t/n$.
Let the (multi)-subset of rows obtained be denoted by $B$ and the matrix formed
from the rows of $B$ be denoted $\hat{A}$.
\eat{Since the rows are sampled i.i.d., the expected size of $B$ is $t$.}
For every $i \in B$, define the indicator random variable $X_i$ which is
$1$ if and only if the randomly sampled index $i$ \blit{satisfies
$A_i^C = b$, which occurs with probability $|T|/n$.}
\eat{is included in the random sample, and is in the
set $T$.
That is, $t_i = 1$ if and only if $i \in \hat{T} := T \cap B$.}
Next, we define $\hat{T} = T \cap B$ so that
$|\hat{T}| = \sum_{i = 1}^t X_i$ and the estimator
$Z = \frac{n}{t}|\hat{T}|$ has $\E(Z) = |T|$.
Finally, apply an additive form of the Chernoff bound:
\begin{align*}
\prob \left( |Z - \E(Z)| \ge \eps n \right) &=
\prob \left( \left| \frac{n}{t}|\hat{T}| - |T| \right| \ge \eps n  \right) \\
      &= \prob \left( \left| |\hat{T}| - \frac{t}{n}|T| \right| \ge \eps t \right) \\
      &\le 2 \exp \left(-\eps^2 t \right).
\end{align*}
Setting $\delta = 2 \exp \left(-\eps^2 t \right)$ allows us to choose
$t = O(\eps^{-2} \log(1/\delta))$, which is independent of $n$ and $d$.
\eat{Treating $\eps$ and $\delta$ as constants, this space bound is also constant.}
The final bound comes from observing that $\|f\|_1 = n, f_{e(b)} = |T| $ and
$\hat{f}_{e(b)} = Z$.
\end{proof}

\newpage
\subsection{Omitted Proof for Section \ref{sec: Fp-estimation}}
\label{sec: Fp-estimation-bound}
A key step in the proof of Theorem \ref{thm: F_p-estimation-lower-bound} is that in 
Equation \eqref{eq: Fp-value-to-bound}, the expression
\[ 2^{cdp + \epsilon dp + \Theta( (1-p)cd )}\cdot O(d^{1-p}) \]
can be bounded by a manageable power of two.
We formalize this in Lemma \ref{lem: c-bound-Fp}.

\begin{lemma} \label{lem: c-bound-Fp}
Under the same assumptions as in Theorem \ref{thm: F_p-estimation-lower-bound},
there exists
a small constant $c > 0$ which bounds Equation \eqref{eq: Fp-value-to-bound} by at most
$2^{(1- \alpha)\epsilon d}$ for some $\alpha > 0$.
\end{lemma}

\begin{proof}
Here we use base-$2$ logarithms and let $0 < c < 1$ be a small constant which we
need to bound.
Also, let $0 < p < 1$ be a given constant.
Observe that the $O(d^{1-p})$ term only contributes positively in the
exponent term of \eqref{eq: Fp-value-to-bound}
so we can ignore it from the calculation.
Write $2^{\Theta(cd(1-p))} = 2^{cd \alpha (1-p)}$ for $\alpha > 0$.
This follows from:
\begin{equation}
  {d \choose cd} \le \left(\frac{e d}{cd}\right)^{cd}
                 \le 2^{(2 + \log \frac{1}{c})cd} \label{eq: alpha-bound}
\end{equation}
so let $\alpha = 2 + \log \frac{1}{c}$.
For clarity, we proceed by using the trivial identity $1 - (1 - \nu) = \nu$ and show that
$1 - \nu > 0$ for $\nu$ a function of $c,p,d$.
We need to ensure:
\begin{equation}
cpd + {\epsilon dp} + \alpha cd(1-p) \le (1-\alpha){\epsilon d}.
\end{equation}
This amounts to showing that:
\[ \nu \triangleq cp/\epsilon + p + \alpha c(1-p)/\epsilon \le (1-\alpha)\]
Now, $\nu = p(c/\epsilon + 1 -  \alpha c/\epsilon) + \alpha c/\epsilon$ and we require $\nu < 1$.
We may enforce the weaker property of $p(c/\epsilon + 1 -  \alpha/\epsilon) < 1$ because
$c > 0$ and for $c < 4$ we also have $\alpha > 0$ (inspection on Equation
\eqref{eq: alpha-bound}) so $ \alpha c/\epsilon > 0$, and so can be omitted.
Solving for $c$ we obtain $c(1-\alpha) < \epsilon (1/p - 1)$.
Recalling the definition of $\alpha$ this becomes:
\begin{equation}
c(\log c - 1) < \epsilon(1/p - 1)
\end{equation}
from which positivity on $c$ yields $c \log c < \epsilon(1/p-1)$.
Hence, it is enough to use $c < \epsilon(1/p - 1)$.
\end{proof}

\eat{
\section{Upper Bound for Fixed Query Sizes}

For the general problem, it suffices to use compact summary algorithms
for each possible column set $C$.
Such sketches are known to exist for many of our problems of interest,
such as $F_0$ and other frequency moments.
The number of sketches needed can be bounded directly as follows.

\begin{theorem}
Let $A \in \{ 0,1 \}^{n \times d}$ be the instance for the projected
problem.
Fix a query size $k$ and let $C$ be a set of queries.
Then the number of sketches required to approximate the projected
problem is at most $(ed/k)^k$ if $C$ contains queries of size at most $k$.
\end{theorem}

\begin{proof}
If $C$ contains queries of size at most $k$ then we must
generate a sketch for all subsets of $[d]$ up to and including cardinality $k$.
Therefore, the number of sketches necessary is:
\begin{equation}
  \sum_{i=0}^k {d \choose i} \le \left(\frac{e d}{k}\right)^{k}.
\end{equation}

\end{proof}
}

\end{document}